\documentclass[11pt]{llncs}
\usepackage[smallbib]{}
\usepackage{makeidx}

\frontmatter          

\usepackage{amssymb,amsfonts,amsthm}
\usepackage{amsmath,epsfig,url,graphicx,xspace}
\usepackage{algorithm}
\usepackage{algorithmic}
\usepackage{fullpage}
\newcommand{\expect}{{\mathbf{E}}}

\renewcommand{\algorithmiccomment}[1]{/* #1 */}

\newcommand{\groundset}{{\mathcal{U}}}
\newcommand{\indsets}{{\mathcal{I}}}

\newcommand{\val}{{\operatorname{val}}}

\begin{document}

\title{Randomized Online Algorithms for the Buyback Problem}

\author{Ashwinkumar B.V. \and Robert Kleinberg}
\institute{Cornell University, Ithaca, NY \\ \email{\{ashwin85,rdk\}@cs.cornell.edu}}

\maketitle

\begin{abstract}
In the matroid buyback problem,
an algorithm observes a sequence of bids and 
must decide whether to accept each bid at the moment it arrives,
subject to a matroid constraint on the set of accepted bids.
Decisions to reject bids are irrevocable, whereas decisions
to accept bids may be canceled at a cost which is a fixed 
fraction of the bid value.  
We present a new randomized algorithm for this problem, and we 
prove matching upper and lower bounds to establish that the 
competitive ratio of this algorithm, against an oblivious 
adversary, is the best possible.  We also observe that when 
the adversary is adaptive, no randomized algorithm can improve
the competitive ratio of the optimal deterministic algorithm.
Thus, our work completely resolves the question of what
competitive ratios can be achieved by randomized algorithms
for the matroid buyback problem.
\end{abstract}

\section{Introduction}
\label{sec:intro}

Imagine a seller allocating a limited inventory
(e.g. impressions of a banner ad on a specified website
at a specified time in the future) to a sequence of 
potential buyers who arrive sequentially, submit
bids at their arrival time, and expect allocation
decisions to be made immediately after submitting
their bid.  An informed seller who 
knows the entire bid sequence can achieve much higher
profits than an uninformed seller who discovers the
bids online, because of the possibility that a very
large bid is received after the uninformed seller has
already allocated the inventory.  A number of recent
papers~\cite{BHK-ec09,CFMP-soda09} have proposed a 
model that offsets this possibility by allowing the
uninformed seller to cancel earlier allocation decisions,
subject to a penalty which is a fixed fraction of the canceled bid 
value.  This option of canceling an allocation and
paying a penalty is referred to as \emph{buyback}, and
we refer to online allocation problems with a buyback
option as \emph{buyback problems}.

Buyback problems have both theoretical and practical 
appeal.  In fact, Babaioff et al.~\cite{BHK-ec09} report
that this model of selling was described to them by the
ad marketing group at a major Internet software company.
Constantin et al.~\cite{CFMP-soda09} cite numerous other
applications including allocation of TV, radio, and newsprint
advertisements; they also observe that advance booking with
cancellations is a common practice in the airline industry, 
where limited inventory is oversold and then, if necessary,
passengers are ``bumped'' from flights and compensated with a penalty
payment, often in the form of credit for future flights.

Different buyback problems are distinguished from each other
by the constraints that express which sets of bids can be
simultaneously accepted.  In the simplest case, the only
constraint is a fixed upper bound on the total number of
accepted bids.  Alternatively, there may be a bipartite 
graph whose two vertex sets are called \emph{bids} and
\emph{slots}, and a set of bids may be simultaneously
accepted if and only if each bid in the set can be 
matched to a different slot using edges
of the bipartite graph.  Both of these examples are
special cases of the \emph{matroid buyback problem},
in which there is a matroid structure on the bids,
and a set of bids may be simultaneously accepted
if and only if they constitute an independent set 
in this matroid.  Other types of constraints (e.g.~knapsack
constraints) have also been studied in the context of buyback
problems~\cite{BHK-ec09}, but the matroid buyback problem
has received the most study.  This is partly because of its desirable
theoretical properties --- the offline version of the problem
is computationally tractable, and the online version admits
an online algorithm whose payoff is identical to that of the
omniscient seller when the buyback penalty is zero --- and
partly because of its well-motivated special cases, such as
the problem of matching bids to slots described above.

As is customary in the analysis of online algorithms, we 
evaluate algorithms according to their competitive
ratio: the worst-case upper bound on the ratio between
the algorithm's (expected) payoff and that of an informed 
seller who knows the entire bid sequence and always allocates
to an optimal feasible subset without paying any penalties.
The problem of deterministic matroid buyback algorithms has
been completely solved: a simple algorithm was proposed and
analyzed by Constantin et al.~\cite{CFMP-workshop,CFMP-soda09} 
and, independently, Babaioff et al.~\cite{BHK-workshop}, and
it was recently shown~\cite{BHK-ec09} that the competitive 
ratio of this algorithm is optimal for deterministic matroid 
buyback algorithms, even for the case of rank-one matroids 
(i.e., selling a single indivisible good).  However, this
competitive ratio can be strictly improved by using a 
randomized algorithm against an oblivious adversary.  
Babaioff et al.~\cite{BHK-ec09} showed that this result
holds when the buyback penalty factor is sufficiently small, 
and they left open the question of determining the optimal
competitive ratio of randomized algorithms --- or even 
whether randomized algorithms can improve on the competitive
ratio of the optimal deterministic algorithm when the buyback
factor is large.  

Our work resolves this open question by supplying a randomized
algorithm whose competitive ratio (against an oblivious adversary)
is optimal for all values of the buyback penalty factor.  We 
present the algorithm and the upper bound on its competitive 
ratio in Section~\ref{s:algo} and the matching lower bound
in Section~\ref{s:lowerbound}.  Our algorithm is also much
simpler than the randomized algorithm of~\cite{BHK-ec09},
avoiding the use of stationary renewal processes.  It may be
viewed as an online randomized reduction that transforms an
arbitrary instance of the matroid buyback problem into a 
specially structured instance on which deterministic algorithms
are guaranteed to perform well.  Our matching lower bound
relies on defining and analyzing a suitable continuous-time 
analogue of the single-item buyback problem.

\paragraph{Adaptive adversaries.}  In this paper we analyze 
randomized algorithms with an oblivious adversary.  If the
adversary is adaptive\footnote{A distinction between 
\emph{adaptive offline} and \emph{adaptive online} adversaries
is made in~\cite{RequestAnswer,BorodinElYaniv}.  When we refer
to an adaptive adversary in this paper, we mean an adaptive
offline adversary.}, then no randomized algorithm can 
achieve a better competitive ratio than that achieved by
the optimal deterministic algorithm.  This fact is a direct
consequence of a more general theorem asserting the same
equivalence for the class of \emph{request answer games}
(Theorem 2.1 of~\cite{RequestAnswer} or 
Theorem 7.3 of~\cite{BorodinElYaniv}), a class of 
online problems that includes the buyback problem.\footnote{The
definition of request answer games in~\cite{BorodinElYaniv} 
requires that the game must have a minimization objective,
whereas ours has a maximization objective.  However, the
proof of Theorem 7.3 in~\cite{BorodinElYaniv} goes through,
with only trivial modifications, for request answer games
with a maximization objective.}

\paragraph{Strategic considerations.}  In keeping 
with~\cite{BHK-workshop,BHK-ec09}, we treat the
buyback problem as a pure online optimization
with non-strategic bidders.  For an examination
of strategic aspects of the buyback problem, we 
refer the reader to~\cite{CFMP-soda09}.


\paragraph{Related work.}  The buyback model was first investigated by 
Constantin et al.~\cite{CFMP-workshop,CFMP-soda09}
and Babaioff et al.~\cite{BHK-workshop,BHK-ec09}.
The optimal deterministic algorithm for the matroid
buyback problem was presented 
in~\cite{BHK-workshop,CFMP-workshop,CFMP-soda09}
and a proof of its optimality appeared
in~\cite{BHK-workshop,BHK-ec09}.
Constantin et al.~also investigated strategic aspects 
of the matroid buyback problem in~\cite{CFMP-workshop,CFMP-soda09};
this research was featured in a recent survey of theory research
at Google in ACM SIGACT News~\cite{TRAG}.
Babaioff et al.~presented algorithms for the knapsack buyback 
problem~\cite{BHK-workshop,BHK-ec09} and designed a 
randomized algorithm for the matroid buyback problem
that strictly improves the competitive ratio of the 
optimal deterministic algorithm when the adversary 
is oblivious and the buyback penalty factor is sufficiently 
small~\cite{BHK-ec09}.  

Prior to the aforementioned work on the buyback problem,
several earlier papers considered models in which allocations,
or other commitments, could be cancelled at a cost.  
Biyalogorsky et al.~\cite{BCFG99} studied such ``opportunistic
cancellations'' in the setting of a seller allocating
$N$ units of a good in a two-period model, demonstrating 
that opportunistic cancellations could improve allocative
efficiency as well as the seller's revenue.  Sandholm
and Lesser~\cite{SandholmL01} analyzed a more general 
model of ``leveled commitment contracts'' and proved
that leveled commitment never decreases the expected payoff
to either contract party.  However, to the best of our knowledge,
the buyback problem studied in this paper and its direct
precursors~\cite{BHK-workshop,BHK-ec09,CFMP-workshop,CFMP-soda09}
is the first to analyze commitments with cancellation costs
in the framework of worst-case competitive analysis rather
than average-case Bayesian analysis.

\section{Preliminaries}\label{s:preliminaries}
First we define the problem in the setting of single item and then 
generalize the definition in the case of matroids. 
\subsection{Single Item Case}
The seller has a single item to allocate. The bids $v_1,v_2,\ldots,v_n$
come in a sequence and when bid $v_i$ arrives the seller must either commit 
or reject the bid immediately. When the seller commits, the previous 
commitment must be revoked by paying a penalty of $f\cdot v_j$, where 
$v_j$ is the bid being revoked and $f \geq 0$ is a fixed number called 
the \emph{buyback factor}.  This implies that at the end of processing 
the bid sequence, the seller's payoff is equal to the final committed 
bid minus $f$ times the sum of all revoked bids. 
The customer with the final accepted bid gets the item.

\subsection{General model for matroids}
Consider a matroid\footnote{See~\cite{oxley} for the definition 
of a matroid.} $(\groundset,\indsets)$ where $\groundset$ is 
the ground set and $\indsets$ is the set of independent subsets 
of $\groundset$. 
We describe the problem abstractly and then relate it to the single item 
case.  We will assume that the ground set $\groundset$ is identified
with the set $\{1,\ldots,n\}$.  There is a bid value $v_i \geq 0$ associated
to each element $i \in \groundset$.  The information available to the 
algorithm at time $k \; (1 \leq k \leq n)$ consists of the first $k$
elements of the bid sequence --- i.e. the subsequence $v_1,v_2,\ldots,v_k$ ---
and the restriction of the matroid structure to the first $k$
elements.  (In other words, for every subset $S \subseteq \{1,2,\ldots,k\}$,
the algorithm knows at time $k$ whether $S \in \indsets.$)

At any step the algorithm can choose a subset 
$S^{k}\subseteq S^{k-1}\cup \{k\}$.  This set $S^k$ must be an independent 
set, i.e $S^k\in \mathcal{I}$.  Hence the final set held by the algorithm 
is $R=S^n$.  The algorithm must perform a buyback for every element of
$B= \left( \cup^n_{i=1}S^i \right) \backslash S^n$. 
For any set $S\subseteq \groundset$ let $\val(S)=\sum_{i\in S}v_i$. 
Finally we define the payoff of the algorithm as $\val(R)-f\cdot \val(B)$. 
This definition generalizes the single item case, which corresponds
to the case in which $\indsets$ consists of all one-element
subsets of $\groundset$.

\section{Randomized algorithm against oblivious adversary} \label{s:algo}

In this section we give a randomized algorithm with competitive ratio 
$- W \left( \frac{-1}{e  (1+f)} \right)$ 
against an oblivious 
adversary.  Here $W$ is Lambert's $W$ function\footnote{Lambert's $W$ function is multivalued for our
domain. We restrict to the case where $W \left( \frac{-1}{e  (1+f)} \right)\leq-1$.}, defined as the inverse 
of the function $z \mapsto z e^z$. The design of our randomized 
algorithm is based on two insights: 
\begin{enumerate}
\item  Although the standard greedy online algorithm for picking 
a maximum-weight basis of a matroid can perform arbitrarily poorly
on a worst-case instance of the buyback problem, it performs well
when the ratios between values of different matroid elements are 
powers of some scalar $r > 1+f$.  (We call such instances
``$r$-structured.'')
\item  There is a randomized reduction from arbitrary instances of 
the buyback problem to instances that are $r$-structured.
\end{enumerate}

\subsection{The greedy algorithm and $r$-structured instances}
\label{s:greedy}

\newcommand{\GMA}{{\mathsf{GMA}}}
\newcommand{\ALG}{{\mathsf{ALG}}}
\newcommand{\Filter}[1]{{\mathsf{Filter}(#1)}}
\newcommand{\RA}{{\mathsf{RandAlg}}}

\begin{definition}  Let $r>1$ be a constant.
An instance of the matroid buyback problem is
\emph{$r$-structured} if for every pair of 
elements $i,j,$ the ratio $v_i/v_j$ is equal 
to $r^l$ for some $l \in \mathbb{Z}.$
\end{definition}

\begin{algorithm}[h]
\caption{Greedy Matroid Algorithm ($\GMA$)}
\label{alg:matroid}
\begin{algorithmic}[1]
\STATE Initialize $S = \emptyset.$
\FORALL{elements $i$, in order of arrival,}
\IF{$S\cup \{i\}\in \indsets$} \label{algstep:indep-test}
\STATE Sell to $i$. \label{step:greedy-a}
\ELSE
\STATE Let $j$ be an element of
smallest value such that  $S\cup \{i\}\setminus \{j\}\in \indsets$.
\label{algstep:eprime}
\IF{$v_j < v_i$}
\STATE Sell to $i$ and buy back $j$. \label{step:greedy-b}
\ENDIF
\ENDIF
\ENDFOR
\end{algorithmic}
\end{algorithm}

\begin{lemma} \label{lem:greedy}
For $r > 1+f$, when the greedy matroid algorithm is executed 
on an $r$-structured instance of the matroid buyback problem, 
its competitive ratio is at most $\frac{r-1}{r-1-f}.$
\end{lemma}
\begin{proof}
As is well known, at termination the set $S$ selected by $\GMA$
is a maximum-weight basis of the matroid. 
To give an upper bound on the total buyback payment, we 
define a set $B(i)$ for each $i \in \groundset$ recursively
as follows: if $\GMA$ never sold to $i$, or sold to $i$
in step~\ref{step:greedy-a}, then $B(i) = \emptyset.$  If $\GMA$
sold to $i$ in step~\ref{step:greedy-b} while 
buying back $j$, then $B(i) = \{j\} \cup B(j).$  By induction
on the cardinality of $B(i),$ we find that the set 
$\{v_x/v_i \,|\, x \in B(i)\}$ consists of distinct 
negative powers of $r$, so 
\[
\sum_{x \in B(i)} v_x \leq v_i \cdot \sum_{i=1}^{\infty} r^{-i} 
= \frac{v_i}{r-1}.
\]
By induction on the number of iterations of the
main loop, the set $\bigcup_{i \in S} B(i)$ consists of all the
elements ever bought back by $\GMA$; consequently, the
total buyback payment is bounded by
\[
f \cdot \sum_{i \in S} \sum_{x \in B(i)} v_x \leq \frac{f}{r-1} \sum_{i \in S} v_i.
\]
Thus, the algorithm's net payoff is at least $1 - \frac{f}{r-1}$ times 
the value of the maximum weight basis.
\end{proof}

\subsection{The random filtering reduction} \label{s:reduction}

Consider two instances $\mathbf{v},\mathbf{w}$
of the matroid buyback problem, consisting of
the same matroid $(\groundset,\indsets)$, with
its elements presented in the same order, but
with different values: element $i$ has values
$v_i,\, w_i$ in instances $\mathbf{v},\mathbf{w}$,
respectively.  Assume furthermore that $v_i \geq w_i$
for all $i$, and that both values $v_i,w_i$ are revealed
to the algorithm at the time element $i$ arrives.  
Given a (deterministic or randomized) algorithm $\ALG$ which
achieves expected payoff $P$ on instance $\mathbf{w}$,
we present here an algorithm $\Filter{\ALG}$
which achieves expected payoff $P$ on instance
$\mathbf{v}.$

\begin{algorithm}[h]
\caption{Random Filtering Algorithm $\Filter{\ALG}$}
\label{alg:matroid}
\begin{algorithmic}[1]
\STATE Initialize $S = \emptyset.$
\FORALL{elements $i$, in order of arrival,}
\STATE Observe $v_i,w_i.$
\STATE Randomly sample $x_i = 1$ with probability $w_i/v_i$, else $x_i=0$.
\STATE Present element $i$ with value $w_i$ to $\ALG.$
\IF{$\ALG$ sells to $i$ {\bf and} $x_i=1$} 
\STATE Sell to $i$.
\ENDIF
\IF{$\ALG$ buys back an element $j$ {\bf and} $x_j=1$}
\STATE Buy back $j$.
\ENDIF
\ENDFOR
\end{algorithmic}
\end{algorithm}

\begin{lemma}  The expected payoff of $\Filter{\ALG}$ on instance
$\mathbf{v}$ equals the expected payoff of $\ALG$ on instance $\mathbf{w}$.
\label{lem:filter}
\end{lemma}
\begin{proof}
For each element $i \in \groundset,$ let $\sigma_i = 1$ if $\ALG$ 
sells to $i$, and let $\beta_i=1$ if $\ALG$ buys back $i$.  Similarly,
let $\sigma'_i = 1$ if $\Filter{\ALG}$ sells to $i$, and let 
$\beta'_i=1$ if $\Filter{\ALG}$ buys back $i$.  Observe that
$\sigma'_i = \sigma_i x_i$ and $\beta'_i = \beta_i x_i$ for all
$i \in \groundset,$ and that the random variable $x_i$ is 
independent of $(\sigma_i,\beta_i).$  Thus,
\begin{align*}
\expect \left[ \sum_{i \in \groundset} \sigma'_i v_i 
- (1+f)\beta'_i v_i \right] 
&=
\expect \left[ \sum_{i \in \groundset} \sigma_i x_i v_i 
- (1+f) \beta_i x_i v_i \right] \\
&= 
\sum_{i \in \groundset} \expect[\sigma_i - (1+f) \beta_i] \expect[x_i v_i] 
\\ &=
\sum_{i \in \groundset} \expect[\sigma_i - (1+f) \beta_i] w_i 
\\ &=
\expect \left[ \sum_{i \in \groundset} \sigma_i w_i - (1+f) \beta_i w_i \right].
\end{align*}
The left side is the expected payoff of $\Filter{\ALG}$ on instance
$\mathbf{v}$ while the right side is the expected payoff of $\ALG$ 
on instance $\mathbf{w}$.
\end{proof}

\subsection{A randomized algorithm with optimal competitive ratio}
\label{s:optimal}

In this section we put the pieces together, to obtain a randomized
algorithm with competitive ratio $- W \left(\frac{-1}{e (1+f)} \right)$
against oblivious adversary\footnote{Note that the algorithm is written in an offline
manner just for convenience and can be implemented as an online algorithm}.

\begin{algorithm}[h]
\caption{Randomized Algorithm $\RA(r)$}
\label{alg:matroid}
\begin{algorithmic}[1]
\STATE {\bf Given:} a parameter $r > 1+f$.
\STATE Sample $u \in [0,1]$ uniformly at random.
\FORALL{elements $i$}
\STATE Let $z_i = u + \lfloor \ln_r(v_i) - u \rfloor.$
\STATE Let $w_i = r^{z_i}.$
\ENDFOR
\STATE Run $\Filter{\GMA}$ on instances $\mathbf{v,w}$.
\end{algorithmic}
\end{algorithm}

\begin{lemma} \label{lem:rounding}
For all $i \in \groundset$, we have $v_i \geq w_i$ and
$\expect[w_i] = \frac{r-1}{r \ln(r)} v_i.$
\end{lemma}
\begin{proof}
The random variable $\ln_r(v_i) - z_i$ is equal to the 
fractional part of the number $\ln_r(v_i) - u,$ which
is uniformly distributed in $[0,1]$ since $u$ is 
uniformly distributed in $[0,1].$  It follows that 
$w_i / v_i$ has the same distribution as $r^{-u}$, which
proves that $v_i \geq w_i$ and also that
$$
\expect \left[ \frac{w_i}{v_i} \right] = \int_0^1 r^{-u} \, du =
\left. - \frac{1}{\ln(r)} \cdot r^{-u} \right|_0^1 =
\frac{r-1}{r \ln(r)}.
$$
\end{proof}

\begin{theorem} \label{thm:rand-alg}
The competitive ratio of $\RA$ against an oblivious adversary
is $\frac{r \ln(r)}{r-1-f}.$
\end{theorem}
\begin{proof}
Let $S^* \subseteq \groundset$ denote the maximum-weight basis
of $(\groundset,\indsets)$ with respect to the weights 
$\mathbf{v}$.  Since the mapping from $v_i$ to $w_i$ is
monotonic (i.e., $v_i \geq v_j$ implies $w_i \geq w_j$),
we know that $S^*$ is also a maximum-weight basis of $(\groundset,\indsets)$
with respect to the weights $\mathbf{w}$\footnote{There may be other maximum-weight basis of $\mathbf{w}$
which were not maximum-weight basis of $\mathbf{v}$.}.  Let 
$v(S^*) = \sum_{i \in S^*} v_i$ and let $w(S^*) = \sum_{i \in S^*} w_i.$

The input instance $\mathbf{w}$ is $r$-structured, so the
payoff of $\GMA$ on instance $\mathbf{w}$ is at least 
$\frac{r-1-f}{r-1} w(S^*)$.  
The modified weights $w_i$ satisfy two properties that
allow application of algorithm $\Filter{\ALG}$: the
value of $w_i$ can be computed online when $v_i$ is revealed
at the arrival time of element $i$, and it satisfies 
$w_i \leq v_i.$  By Lemma~\ref{lem:filter}, the expected
payoff of $\Filter{\GMA}$ on instance $\mathbf{v}$,
conditional on the values $\{w_i : i \in \groundset\}$, is
at least $\left(\frac{r-1-f}{r-1}\right) \cdot w(S^*)$.  Finally, by 
Lemma~\ref{lem:rounding} and linearity of expectation,
$ \expect \left[ w(S^*) \right] \geq
\left( \frac{r-1}{r \ln(r)} \right) \cdot v(S^*).$
The theorem follows by combining these bounds.
\end{proof}

The function $f(r) = \frac{r \ln(r)}{r-1-f}$ on the interval
$r \in (1+f,\infty)$ is minimized when 
$ - \frac{r}{1+f} = W \left( \frac{-1}{e (1+f)} \right)$ 
and $ f(r) = - W \left(\frac{-1}{e (1+f)} \right)$.
This completes our analysis of the randomized algorithm
$\RA(r)$.

\section{Lower Bound}\label{s:lowerbound}

We prove the lower bound on the competitive ratio of randomized
algorithms for online algorithms with buyback against an oblivious
adversary. The proof is by first reducing to a continuous version of
the problem and then applying Yao's Principle~\cite{YaoPrinciple}. 
As noted in the introduction, the lower bound in the case of an
adaptive adversary matches the lower bound for deterministic 
algorithms.  Both of these lower bounds are for the single item 
case and hence are also applicable for the general matroid case.

\subsection{Reduction to continuous version}
Consider the following continuous version of the problem for the single 
item case. Time starts at $t=1$ and stops at some time $t=x$. The value 
of $x$ is not known to the algorithm. The algorithm at any instant in 
time can make a mark. The final payoff of the algorithm is equal to the 
time at which it made its final mark minus $f$ times the sum of times of
marks before the final mark.  There is a clear relationship between
this problem and the single item buyback problem.  In particular, we 
can transform any algorithm for the single item buyback problem with 
competitive ratio $c$ to an algorithm for the continuous case with 
competitive ratio $c\times (1+\epsilon)$ for arbitrarily small $\epsilon > 0$. 
This transformation works by running the single item buyback algorithm
on the input sequence $1,1+\delta,(1+\delta)^2,(1+\delta)^3,\ldots$
for sufficiently small $\delta>0$, and making marks at the times 
$t$ corresponding to the values accepted in the execution of the single 
item buyback algorithm.

\subsection{Lower bound against oblivious adversaries}

\begin{theorem}
Any randomized algorithm for the continuous version of the single item buyback problem has competitive ratio at least $- W \left( \frac{-1}{e  (1+f)} \right)$.
\end{theorem}

The proof is an application of Yao's Principle~\cite{YaoPrinciple}. We
give a one-parameter family of input distributions (parametrized by a
number $y>1$) for the continuous version and prove that
any deterministic algorithm for the continuous version of the problem
must have a competitive ratio which tends to $- W \left( \frac{-1}{e
    (1+f)} \right)$ as $y \rightarrow \infty$. 
It is easy to note that an input to the
continuous version is completely specified by the time $x$
at which the input stops, and hence the input
distribution is just a distribution on $x$. For a 
given $y>1$, let the probability density for the stopping times be
defined as follows.
\begin{eqnarray}
\textbf{f}(x)=&1/x^2 &\textrm{if }x< y \nonumber \\
\textbf{f}(x)=&0 &\textrm{if }x>y
\end{eqnarray}
Note that the above definition is not a valid probability density
function, so we place a point mass at $x=y$ of probability
$\frac{1}{y}$. Hence our distribution is a mixture of discrete and
continuous probability. For notational convenience let $d(F(x))=\textbf{f}(x)$
where $F$ is the cumulative distribution function. Also let
$G(x)=1-F(x)$. Any deterministic algorithm is defined by a set 
$T=\{u_1,u_2,\ldots,u_k\}$ of times at which it makes a mark(Given that it does not
stop before that time).


\begin{lemma} \label{geometric}
There exists an optimal deterministic algorithm for the distribution described by $T=\{1,w,w^2,\ldots,w^{k-1}\}$ for some w,k.
\end{lemma}
\begin{proof}
Let $T=\{u_1,u_2,\ldots,u_k\}$. We prove that $u_i=u_{i+1}^{(i-1)/i}$ for
$i\in [k-1]$ by induction and it is easy to see that the claim 
follows from this. For lack of space we just prove the inductive
case. Please refer to the appendix for the base case. Let $u_0=0$
and $u_{k+1}=\infty$.

It is easy to see that the algorithm's expected payoff, $P$, is 
$\underset{i=1}{\overset{k}{\sum}}\int_{u_i}^{u_{i+1}} \! (u_i-f \cdot \underset{j=1}{\overset{i-1}{\sum}}u_j)\, d(F(y))$. We simplify this expression as follows.
\begin{eqnarray}
P=&\underset{i=1}{\overset{k}{\sum}}\int_{u_i}^{u_{i+1}} \! (u_i-f \cdot \underset{j=1}{\overset{i-1}{\sum}}u_j) \, d(F(y)) \nonumber \\
=&\underset{i=1}{\overset{k}{\sum}}\int_{u_i}^{\infty} \! (u_i-(1+f)\cdot u_{i-1}) \, d(F(y)) \nonumber \\
=&\underset{i=1}{\overset{k}{\sum}} (u_i-(1+f)\cdot u_{i-1})\cdot G(u_i) \label{simplifiedyao}
\end{eqnarray}
Now we rewrite this equation to express the right side as a function of $u_i$, using $\rho_i$ to denote the sum of all terms on the right side except for the $i,i+1$ terms.  Crucially, $\rho_i$ is independent of $u_i$.
\begin{eqnarray}
P=&(u_i-(1+f)\cdot u_{i-1})\cdot G(u_i)+(u_{i+1}-(1+f)\cdot u_i)\cdot G(u_{i+1})+ \rho_i \nonumber \\
=&(u_i-(1+f)\cdot u_{i-1})\cdot \frac{1}{u_i}+(u_{i+1}-(1+f)\cdot u_i)\cdot \frac{1}{u_{i+1}}+ \rho_i
\end{eqnarray}

If we differentiate $P$ with respect to $u_i$, equate to 0, and solve, then we obtain the equation $u_i^2=u_{i-1}\cdot u_{i+1}$. By induction we know that $u_{i-1}=u_i^{(i-2)/(i-1)}$. Substituting and solving we get the necessary equation. 
\end{proof}

\begin{lemma}
For any algorithm described by $T=\{1,w,w^2,\ldots,w^{k-1}\}$, the competitive ratio is bounded below by a number which tends to $- W \left( \frac{-1}{e  (1+f)} \right)$ as y tends to $\infty$.
\end{lemma}
\begin{proof}
It is easy to see that the expected payoff, $V$, of a prophet who knows the stopping time $x$ is given by the following equation.
\begin{eqnarray}
V=\int_1^y \! \frac{1}{x^2}\cdot x \, dx+\frac{1}{y}\cdot y=1+\ln(y)
\end{eqnarray}
Now we compute the payoff for any algorithm described by $T=\{1,w,w^2,\ldots,w^{k-1}\}$.
\begin{eqnarray}
P=&1\cdot G(1)+\underset{i=1}{\overset{k-1}{\sum}}(w^i-(1+f)w^{i-1})\cdot G(w^i) \nonumber \\
=&1\cdot 1+\underset{i=1}{\overset{k-1}{\sum}}(w^i-(1+f)w^{i-1})\cdot \frac{1}{w^i} \nonumber \\
=&1+(k-1)\cdot \frac{w-1-f}{w} 
\end{eqnarray}
Hence if $c$ is the competitive ratio we have that.
\begin{eqnarray}
\frac{1}{c}=\frac{P}{V} \nonumber 
=&\frac{1+(k-1)\cdot (w-1-f)/w}{1+\ln(y)} \nonumber \\
<& \frac{1}{\ln(y)}+\frac{(k-1)\cdot (w-1-f)/w}{(k-1)\cdot \ln(w)} \nonumber \\
\leq& \frac{1}{\ln(y)}+\underset{u}{\max}(\frac{u-1-f}{u\cdot \ln(u)}) \nonumber \\
\leq&  \frac{1}{\ln(y)}-\frac{1}{W \left( \frac{-1}{e  (1+f)} \right)}
\end{eqnarray}
\end{proof}

{\footnotesize
\bibliographystyle{splncs}
\bibliography{buyback}

\begin{thebibliography}{10}

\bibitem{BHK-ec09}
Babaioff, M., Hartline, J.D., Kleinberg, R.:
\newblock Selling ad campaigns: online algorithms with buyback.
\newblock In: Proc. 10th ACM Conference on Electronic Commerce. (2009)

\bibitem{CFMP-soda09}
Constantin, F., Feldman, J., Muthukrishnan, S., P\'al, M.:
\newblock An online mechanism for ad slot reservations with cancellations.
\newblock In: Proc. 20th Annual ACM-SIAM Symposium on Discrete Algorithms
  (SODA). (2009)  1265--1274

\bibitem{CFMP-workshop}
Constantin, F., Feldman, J., Muthukrishnan, S., P\'al, M.:
\newblock An online mechanism for ad slot reservations with cancellations.
\newblock In: Proc. 4th Workshop on Ad Auctions. (2008)

\bibitem{BHK-workshop}
Babaioff, M., Hartline, J.D., Kleinberg, R.:
\newblock Online algorithms with buyback.
\newblock In: Proc. 4th Workshop on Ad Auctions. (2008)

\bibitem{RequestAnswer}
Ben-David, S., Borodin, A., Karp, R., Tardos, G., Wigderson, A.:
\newblock On the power of randomization in on-line algorithms.
\newblock Algorithmica \textbf{11}(1) (1994)  2--14

\bibitem{BorodinElYaniv}
Borodin, A., El-Yaniv, R.:
\newblock Online Computation and Competitive Analysis.
\newblock Cambridge University Press (1998)

\bibitem{TRAG}
Aggarwal, G., Ailon, N., Constantin, F., Even-Dar, E., Feldman, J., Frahling,
  G., Henzinger, M.R., Muthukrishnan, S., Nisan, N., P\'{a}l, M., Sandler, M.,
  Sidiropoulos, A.:
\newblock Theory research at {G}oogle.
\newblock SIGACT News \textbf{39}(2) (2008)  10--28

\bibitem{BCFG99}
Biyalogorsky, E., Carmon, Z., Fruchter, G.E., Gerstner, E.:
\newblock Research note: Overselling with opportunistic cancellations.
\newblock Marketing Science \textbf{18}(4) (1999)  605--610

\bibitem{SandholmL01}
Sandholm, T., Lesser, V.R.:
\newblock Leveled commitment contracts and strategic breach.
\newblock Games and Economic Behavior \textbf{35} (2001)  212--270

\bibitem{oxley}
Oxley, J.:
\newblock Matroid Theory.
\newblock Oxford University Press (1992)

\bibitem{YaoPrinciple}
Yao, A.C.C.:
\newblock Probabilistic computations: Toward a unified measure of complexity.
\newblock In: Proceedings of the 18th Annual Symposium on Foundations of
  Computer Science, Washington, DC, USA, IEEE Computer Society (1977)  222--227

\end{thebibliography}
}
\appendix
\section{Base case}
We prove here the base case in the inductive hypothesis of proof of lemma \ref{geometric}. Consider the payoff of the algorithm.
\begin{eqnarray}
P=\underset{i=1}{\overset{k}{\sum}} (u_i-(1+f)\times u_{i-1})\times G(u_i) \label{simplifiedyao}
\end{eqnarray}
Similar to the inductive case we rewrite the equation as a function of $u_1$, using $\rho_1$ to denote the sum of all terms on the right side except for the $1^{st},2^{nd}$ terms. 
\begin{eqnarray}
P=&u_1\times G(u_1)+(u_2-(1+f)\times u_1)\times G(u_2)+\rho_1 \nonumber \\
=&u_1\times \frac{1}{u_1}+(u_2-(1+f)\times u_1)\times \frac{1}{u_2}+\rho_1 \nonumber \\
=&-(1+f)\times \frac{u_1}{u_2}+1+1+\rho_1 \nonumber
\end{eqnarray}
It is easy to see that $P$ is a decreasing function of $u_1$.
Hence $u_1=1=u_2^{\frac{1-1}{1}}$.





\end{document}